\documentclass[a4paper,USenglish,cleveref,autoref,thm-restate,cleveref]{lipics-v2021}

\usepackage[utf8]{inputenc}
\usepackage[T1]{fontenc}
\usepackage{amsmath,amsfonts,amsthm,amssymb}
\usepackage{xfrac}					
\usepackage[sanserif,basic]{complexity}

\usepackage[dvipsnames]{xcolor}     
\usepackage{tikz}                   
\usepackage{xspace}

\usepackage{footnote}
\makesavenoteenv{tabular}
\makesavenoteenv{table}

\let\oldsqrt\sqrt
\def\hksqrt{\mathpalette\DHLhksqrt}
\def\DHLhksqrt#1#2{\setbox0=\hbox{$#1\oldsqrt{#2\,}$}\dimen0=\ht0
   \advance\dimen0-0.2\ht0
   \setbox2=\hbox{\vrule height\ht0 depth -\dimen0}%
   {\box0\lower0.4pt\box2}}
\renewcommand\sqrt\hksqrt

\renewcommand{\leq}{\leqslant}
\renewcommand{\geq}{\geqslant}
\renewcommand{\le}{\leqslant}
\renewcommand{\ge}{\geqslant}

\providecommand{\ignore}[1]{} 

\newcommand*{\Otilde}{\widetilde{O}}

\newcommand{\D}{{\cal D}}
\newcommand{\F}{{\cal F}}
\newcommand{\G}{{\cal G}}
\renewcommand{\H}{{\cal H}}

\newcommand{\Y}{{\cal Y}}
\newcommand{\level}{\textsc{Level}}
\newcommand{\m}{{-}}

\newcommand{\kpath}{$k$-path\xspace}
\newcommand{\kpaths}{$k$-paths\xspace}
\newcommand{\ftkpath}{\textsc{$f$-EFT-$k$-Path}\xspace}
\newcommand{\ftkvc}{\textup{$f$-EFT-$k$-VC}\xspace}

\newcommand{\prob}[1]{\mathbb{P} (#1)}

\newcommand*{\nwspace}{\hspace*{.1em}} 

\author{Davide Bilò}%
{Department of Humanities and Social Sciences, University of Sassari, Italy}%
{davidebilo@uniss.it}%
{https://orcid.org/0000-0003-3169-4300} 
{This work was partially supported by the Research Grant FBS2016\_BILO, funded by ``Fondazione di Sardegna'' in 2016.} 

\author{Katrin Casel}%
{Hasso Plattner Institute, University of Potsdam, Germany}%
{katrin.casel@hpi.de}%
{https://orcid.org/0000-0001-6146-8684} 
{} 

\author{Keerti Choudhary}%
{Department of Computer Science and Engineering, Indian Institute of Technology Delhi, India}%
{keerti@iitd.ac.in}%
{https://orcid.org/0000-0002-8289-5930} 
{} 

\author{Sarel Cohen}%
{School of Computer Science, Tel-Aviv-Yaffo Academic College, Israel}%
{sarelco@mta.ac.il}%
{https://orcid.org/0000-0003-4578-1245} 
{} 

\author{Tobias Friedrich}%
{Hasso Plattner Institute, University of Potsdam, Germany}%
{tobias.friedrich@hpi.de}%
{https://orcid.org/0000-0003-0076-6308} 
{} 

\author{J.A.~Gregor Lagodzinski}%
{Hasso Plattner Institute, University of Potsdam, Germany}%
{gregor.lagodzinski@hpi.de}%
{https://orcid.org/0000-0002-8771-1870} 
{} 

\author{Martin Schirneck}%
{Hasso Plattner Institute, University of Potsdam, Germany}%
{martin.schirneck@hpi.de}%
{} 
{} 

\author{Simon Wietheger}%
{Hasso Plattner Institute, University of Potsdam, Germany}%
{simon.wietheger@student.hpi.de}%
{https://orcid.org/0000-0002-0734-0708} 
{} 

\authorrunning{D.~Bilò, K.~Casel, K.~Choudhary, S.~Cohen, T.~Friedrich, M.~Schirneck, and S.~Wietheger}
\Copyright{Davide Bilò, Katrin Casel, Keerti Choudhary, Sarel Cohen, Tobias Friedrich, Martin Schirneck, and Simon Wietheger}

\title{Fixed-Parameter Sensitivity Oracles}



\nolinenumbers 

\hideLIPIcs  

\begin{document}

\maketitle

\begin{abstract}
	We combine ideas from distance sensitivity oracles (DSOs) and fixed-parameter tractability (FPT)
to design sensitivity oracles for FPT graph problems.
An oracle with sensitivity $f$ for an FPT problem $\Pi$ on a graph $G$ with parameter $k$ preprocesses $G$ in time $O(g(f,k) \cdot \poly(n))$. When queried with a set $F$ of at most $f$ edges of $G$, the oracle reports the answer to the $\Pi$\hspace*{.25pt}--\hspace*{.25pt}with the same parameter $k$\hspace*{.25pt}--\hspace*{.25pt}on the graph $G-F$, i.e., $G$ deprived of $F$.
The oracle should answer queries in a time that is significantly faster than merely running
the best-known FPT algorithm on $G-F$ from scratch.

We design sensitivity oracles for the \textsc{$k$-Path} and the \textsc{$k$-Vertex Cover} problem. Our first oracle for \textsc{$k$-Path} has size $O(k^{f+1})$ 
and query time $O(f\min\{f, \log (f) + k\})$.
We use a technique inspired by the work of Weimann and Yuster [FOCS 2010, TALG 2013] on distance sensitivity problems to reduce the space to $O\big(\big(\frac{f+k}{f}\big)^f \big(\frac{f+k}{k}\big)^k fk \cdot \log n\big)$
at the expense of increasing the query time to $O\big(\big(\frac{f+k}{f}\big)^f \big(\frac{f+k}{k}\big)^k f \min\{f,k\} \cdot \log n \big)$. Both oracles can be modified to handle vertex-failures, but we need to replace $k$ with $2k$ in all the claimed bounds.

Regarding \textsc{$k$-Vertex Cover}, we design three oracles offering different trade-offs between the size and the query time.
The first oracle takes $O(3^{f+k})$ space and has $O(2^f)$ query time,
the second one has a size of $O(2^{f+k^2+k})$ and a query time of $O(f{+}k^2)$; finally, the third one takes $O(fk+k^2)$ space and can be queried in time $O(1.2738^k + f)$.
All our oracles are computable in time (at most) proportional to their size and the time needed
to detect a $k$-path or $k$-vertex cover, respectively.
We also provide an interesting connection between \textsc{$k$-Vertex Cover} and the fault-tolerant shortest path problem, by giving a DSO of size $O(\poly(f,k) \cdot n)$  with query time in $O(\poly(f,k))$, where $k$ is the size of a vertex cover.

Following our line of research connecting fault-tolerant FPT and shortest paths problems, 
we introduce parameterization to the computation of distance preservers.
We study the problem, given a directed unweighted graph with a fixed source $s$
and parameters $f$ and $k$,
to construct a polynomial-sized oracle that efficiently reports,
for any target vertex $v$ and set $F$ of at most $f$ edges,
whether the distance from $s$ to $v$ increases at most by an additive term of $k$ in $G-F$.
The oracle size is $O(2^kk^2 \cdot n)$, while the time needed to answer a query is $O(2^kf^\omega k^\omega)$, where $\omega<2.373$ is the matrix multiplication exponent.
The second problem we study is about the construction of bounded-stretch fault-tolerant preservers. We construct a subgraph with $O(2^{fk+f+k} \, k \cdot n)$ edges that preserves those $s$-$v$-distances that do not increase by more than $k$ upon failure of $F$. This improves significantly over the $\widetilde{O}(f n^{2-\frac{1}{2^f}})$ bound in the unparameterized case by Bodwin et al.~[ICALP 2017].
\end{abstract}

\section{Introduction}
\label{sec:intro}

There is a large amount of research in computer science that is devoted to the problems of computing shortest paths and distances in graphs that are subject to a small number of transient failures (see \cite{AhmedBSHJKS20,Bilo21Space-efficient,BCFS21SingleSourceDSO,BodwinDR21, GrandoniVWilliamsFasterRPandDSO, GuRen21ConstructingDSO, Parter15, ParterP13, ParterP16, Ren22Improved, Weimann_Yuster_2013} and the references therein). 
These problems usually ask to compute suitable subgraphs, a.k.a. {\em fault-tolerant spanners}, of the input graph that preserve (approximate) shortest paths even in the presence of edge/vertex failures, or to design data structures, a.k.a {\em distance sensitivity  oracles} (DSOs), that can quickly report such shortest paths when they are queried. In particular DSOs are extremely useful in all those scenarios where we promptly need to retrieve shortest paths in fault-prone networks every time new transient faults are detected. 
As one would expect, depending on the specific problem we are looking for, there are tradeoffs among the time needed by the precomputation algorithm to build the DSO and the time and space required by the DSO to answer to a query.

Most of the research on fault-tolerant spanners and DSOs has focused on the prominent case in which the maximum number of faulty network components observed at any time is a parameter that is independent of the input size.

In this work, we combine ideas from DSOs and fixed-parameter tractability to design sensitivity oracles (SOs) for fixed-parameter tractable (FPT) problems on graphs. More precisely, in the $f$-edge (resp., $f$-vertex) fault-tolerant SO for an FPT problem $\Pi$, we are given an instance $\Pi(G,k)$ of $\Pi$, where $G$ is the $n$-vertex input graph and $k$ is the parameter, and we want to develop a preprocessing algorithm that builds a data-structure (the oracle) that, when queried on a set $F$ of at most $f$ edges (resp., vertices), is able to report the solution of the FPT problem instance $\Pi(G{-}F,k)$, where $G-F$ denotes the graph $G$ deprived of the edges (resp., vertices) in $F$. 


The three important parameters for which we want to find good tradeoffs are (i) the time needed by the query algorithm, (ii) the running time of the preprocessing algorithm, and (iii) the size of the oracle. 
The goal of this paper is to exploit the power of FPT techniques to build fault-tolerant SOs with the following properties. (1) They can answer to queries in a time that is significantly faster than the one required by merely running the best-known FPT algorithm for $\Pi$ on the instance $\Pi(G-F,k)$. (2) Both the running time of the preprocessing algorithm and the size of the oracle must be $O(g(f,k) \cdot \poly(n))$.





To the best of our knowledge Lochet et al.~\cite{lochet} were the first ones who combined ideas from the realm of fixed-parameter tractability to problems related to the computation of fault-tolerant spanners. More precisely, they provided kernels for the problem of computing fault-tolerant $(S \times S)$-reachability spanners in directed graphs, where the size of the kernel depends on the number of faults and other graph parameters, like the independence number and the order of strongly connected components. Recently, Misra~\cite{faulttfvs} considered a fault-tolerant version of the \textsc{Feedback Vertex Set} problem where a single vertex can fail and solved the problem via an FPT algorithm parameterized by the solution size.

Parameterized approaches have been used for the related setting of dynamic algorithms. The crucial difference between dynamic structures and fault tolerant structures is that in the dynamic setting, changes to the input arrive one by one, and there is no bound $f$ on the number of changes. For dynamic structures, one also analyzes update and query time separately. Among the parameterized approaches for dynamic structures, Iwata and Oka~\cite{dynamicfpt2} considered a dynamic setting where both addition and removal of edges are allowed, and addressed problems -- parameterized by the solution size  -- like \textsc{Cluster Vertex Deletion}, \textsc{Feedback Vertex Set}, \textsc{Chromatic Number} and \textsc{Vertex Cover}. In particular for \textsc{Vertex Cover}, they use the Buss rule to only store a subgraph of size in $O(k^2)$ and answer queries independent of $n$ in $O(1.2738^k+k^3)$ by running the best parameterized algorithm for \textsc{Vertex Cover} on this subgraph. For this approach, it is crucial that the input fulfils the promise that updates never create a graph which does not contain a vertex cover bounded by $k$. Alman, Mnich and Vassilevska Williams~\cite{dynamicfpt} improved this result by not requiring this promise. They further studied a plethora of parameterized problems under updates with the objective to give an update/query time with lowest possible dependency on the input size, i.e., possibly only depending on $k$, or at least depending on a factor of the form $f(k)n^{o(1)}$. Among those,  they also considered the \textsc{$k$-Path} problem and give a structure with query time of $O(k!2^{O(k)\log n})$. Chen et al.~\cite{dynamickpath} developed a dynamic graph structure supporting edge insertions and deletions that in allows to recompute a \kpath with amortised update $O(2^{O(k^2)})$ time (plus time for standard graph operations, depending on the data structure). 

\subsection{Our contribution}
\label{subsec:intro_contribution}

We develop $f$-edge/vertex fault-tolerant sensitivity oracles for different FPT problems. The first problem we consider is the well-known \textsc{$k$-Path} problem, which is the decision problem where we want to know whether a (potentially directed) graph $G$ contains a simple path on exactly $k$ vertices (a.k.a.\ a \kpath). 
We denote the edge fault-tolerant variant of this problem with sensitivity $f$ by \ftkpath,
and $f$-VFT-$k$-\textsc{Path} for the variant with vertex failures.
We design two $k$-Path Sensitivity Oracles ($k$-PSOs) for those problems, providing various trade-offs between space requirement\footnote{%
	Throughout this paper, we measure space in the number of machine words on $O(\log n)$ bits,
	where $n$ is the number of vertices in the input graph.
}
of the oracle and its query time. 
In the following, we only provide the main results for the edge-failure scenario. Similar results can be proved for the case in which we consider vertex faults, although the parameter $k$ should then be replaced by $2k$ in all statements.\footnote{We observe that, once we have a solution that handles only edge failures, we can reuse it to handle also vertex failures because of the following graph transformation. We split every vertex $v$ of $G$ into two vertices $v_{\text{in}}$ and $v_{\text{out}}$, where the former inherits all in-edges of $v$ and the latter all out-edges. The failure of vertex $v$ is now equivalent to considering the failure of the edge $(v_{\text{in}},v_{\text{out}})$. This changes the graph size by at most a constant factor of 2, while the parameter $k$ doubles.}

We obtain the following results for the \ftkpath problem
(see Section \ref{sec:k-path} for the technical details).
First we use fault-tolerant lookup trees to design a $k$-Path Sensitivity Oracle whose query time is at most quadratic in $f$ and $k$
and independent of the size of the input graph.
A similar technique has been applied by Chechik et al.~\cite{Chechik17ApproximateSensitiveDistanceOracles}
to design approximate DSOs for multiple edge faults.
The query time of our oracle can be adjusted to the relative sizes of $f$ and $k$.
In every node of the lookup tree, we compute a \kpath as a witness.
For deterministic data structures, we use the algorithm by Tsur~\cite{Tsur_2019} that computes a \kpath in $O(2.554^k \cdot \poly(n))$ time; if we allow randomization, we can use the current-best randomized algorithm of Williams~\cite{Williams_2009} that runs in  $O(2^k \cdot \poly(n))$ time.

\begin{restatable}{theorem}{thm:intro_kpath_tree} \label{thm:intro_kpath_tree}
	There is a deterministic 
	$k$-Path Sensitivity Oracle whose size is $O(k^{f+1})$,
	with a query time of $O(f \cdot \min\{f, \log(f) + k\})$. The deterministic preprocessing algorithm that builds the oracle runs in $O(k^f \, 2.554^k \cdot \poly(n))$ time. If we allow randomization, then the preprocessing algorithm requires $O(k^f \, 2^k \cdot \poly(n))$ time.
\end{restatable}

We note that the time required to construct the deterministic oracle in \autoref{thm:intro_kpath_tree} favorably compares with those achieved for the more general dynamic setting by Alman, Mnich and Vassilevska Williams~\cite{dynamicfpt} (query time of $O(k!2^{O(k)}\log n)$ even for undirected graphs) and Chen et al.~\cite{dynamickpath} ($O(2^{O(k^2)})$ amortized update time).

Furthermore, inspired by the work of Weimann and Yuster~\cite{Weimann_Yuster_2013},
we show that the preprocessing time and space of the data structure
can be significantly reduced for many ranges of $f$ and $k$, at the expense of a larger query time.

\begin{restatable}{theorem}{thm:intro_kpath_subgraphs}
\label{thm:intro_kpath_subgraphs}
	There is a randomized 
	$k$-PSO whose size is $O\big(\big(\frac{f+k}{f} \big)^f \big(\frac{f+k}{k} \big)^k fk \log n \big)$, with a query time of $O\big(\big(\frac{f+k}{f} \big)^f \big(\frac{f+k}{k} \big)^k f \min\{f,k\} \log n\big)$. The time needed by the preprocessing algorithm to build the oracle is $O\big(\big(\frac{f+k}{f} \big)^f \big(\frac{f+k}{k} \big)^k f \nwspace 2^k \nwspace \poly(n)\big)$.
\end{restatable}

\noindent
Note that the expression $\big(\frac{f+k}{f} \big)^f \big(\frac{f+k}{k} \big)^k$ is maximum for $f = k$, where it defaults to $2^{f+k}$.

As it often does, \textsc{$k$-Path} serves as a good test case for a number of related problems. The techniques used for our results in fact can be applied to all problems that search for subgraph characterized by $k$ edges and are in FPT parameterized by $k$. For such a problem $\Pi$, these techniques give an oracle for $\Pi$ under edge failures with size and query time as given in  \autoref{thm:intro_kpath_tree} (or  \autoref{thm:intro_kpath_subgraphs}), where the construction time depends on the best deterministic (or randomized) algorithm to solve $\Pi$. Such an approach can be used for example for \textsc{Graph Motif} (see e.g.~\cite{Pinter2014DeterministicPA}) or \textsc{$k$-Tree} (see e.g.~\cite{Fomin2014EfficientCO}). 
There are also problems that can be reduced to $k$-Path but do not fall into the above category. In the fault tolerant setting, it is not so easy to transfer these reductions in a meaningful way. As a more involved example we consider the  \textsc{Exact Detour} problem (see~\cite{BezakovaCDF19Journal}) where given a graph $G=(V,E)$, two vertices $s, t\in V$ and a parameter $k\in \mathbb N$, the task is to decide if there exists a path of length $d(s,t,G)+k$ in $G$. Our results can be used to derive from \autoref{thm:intro_kpath_tree}  an oracle for a fault tolerant version of \textsc{Exact Detour} that can be built in time $O(n^2\cdot k^{f+1} \cdot 2.2554^k)$ has size $O(n^2\cdot k^{f+2})$ and query time $O(n^2\cdot (k f^2+k^2))$.

The second FPT problem we consider is the famous \textsc{Vertex Cover}, that is, the decision problem in which, given an undirected (resp., directed) input graph $G$ and a positive integer parameter $k$, we want to decide whether $G$ contains a set $C$ of at most $k$ vertices (a.k.a.\ a $k$-vertex cover) such that, for every edge $(u,v)$ of $G$, $C \cap \{u,v\}\neq \emptyset$ (resp, $u \in C$). 
This problem is sometimes referred to as the ``drosophila of FPT'' as virtually every algorithmic technique
in fixed-parameter tractability was demonstrated first using \textsc{Vertex Cover}, 
see the textbooks~\cite{Cygan15ParamertizedAlgorithms,DowneyFellows13Parameterized,Niedermeier06Invitation} for countless examples.
Our work continues this tradition as, in Section \ref{sec:vertex_cover}, we design three $f$-edge fault-tolerant SOs for the \textsc{Vertex Cover} problem with parameter $k$ (\ftkvc).
 The trade-offs we give in our constructions are even broader than in the case of the \textsc{$k$-Path} problem.
We present solutions whose respective sizes and query times range from polynomial to exponential in the parameters $f$ and $k$.
However, in all cases both the space requirement and query time is independent of the graph size.

\begin{restatable}{theorem}{thmintroVCSOthreetof}
\label{thm:intro_VCSO_3_to_f}
	There is a deterministic Vertex Cover Sensitivity Oracle of size $O(3^{f+k})$
	and $O(2^f)$ query time then can be built in time $O(3^{f+k} \, \poly(f,k)+ \poly(n))$.
\end{restatable}

\begin{restatable}{theorem}{thmintroVCSOlowquery}
\label{thm:intro_VCSO_low_query}
	There is a deterministic Vertex Cover Sensitivity Oracle of size $O(2^{f+k^2+k})$
	and $O(f+k^2)$ query time.
	The oracle can be built in time $O(2^{f+k^2+k}(1.2738^k +\poly(f,k))+\poly(n))$.
\end{restatable}

A slight adaption of the dynamic graph structure by Alman, Mnich and Vassilevska Williams~\cite{dynamicfpt} directly gives the following.

\begin{restatable}{corollary}{corintroVCSOlowspace}
\label{cor:intro_VCSO_low_space}
	There is a deterministic Vertex Cover Sensitivity Oracle of size $O(fk+k^2)$
	and $O(1.2738^k + f)$ query time, that can be built in polynomial time. 
\end{restatable}

We conclude Section \ref{sec:vertex_cover} by providing an $f$-edge/vertex fault-tolerant DSO parameterized by the size of a vertex cover of $G$.\footnote{An $f$-edge/vertex fault-tolerant DSO is a data structure that reports the value $d(s,t,G-F)$ when queried on two vertices $s$ and $t$, and a set $F$ of edges/vertices of $G$ of size at most $f$.} More precisely, given a vertex cover $C$ for $G$ with $|C|=k$, we can compute in polynomial time an $f$-DSO of size  $O(n\cdot \poly(k,f))$ with query time in $O(\poly(k,f))$.

So far we applied ideas from fault tolerance to FPT problems,
we conclude our paper with new results in the opposite direction.
We introduce a parameterized variant of the classical fault-tolerant problem of
computing distance preservers.
An $f$-\emph{edge fault-tolerant distance preserver} ($f$-EFT-BFS)\footnote{%
	The abbreviation $f$-EFT-BFS follows a naming convention introduced in~\cite{ParterP13},
	the conference version of \cite{ParterP16},
	and alludes to the origin of this line of research in breadth-first searches.
}
of a directed unweighted graph $G$ with a distinguished source vertex $s$
is a subgraph $H$ of $G$ such that for any vertex $v$ and set $F$ of at most $f$ edges,
we have $d(s,v,G\m F)=d(s,v,H\m F)$. 
Parter~\cite{Parter15} as well as Parter and Peleg~\cite{ParterP16} proved and $\Omega(n^{2-\frac{1}{f+1}})$ space lower bound for $f$-EFT-BFS.
They also constructed a $1$-EFT-BFS with $O(n^{3/2})$ edges.
Bodwin et al.~\cite{BGPW17} generalized this to an $\Otilde(f \nwspace n^{2-\frac{1}{2^f}})$
bound\footnote{%
	The $\Otilde(\cdot)$ notation hides polylogarithmic factors in $n$.
} 
for general $f$.

We investigate whether parameterization can help to sidestep the near-quadratic
lower bound on the size. 
This barrier stems from the requirement to preserve all \emph{replacement distances} $d(s,v,G\m F)$
from $s$ to any $v$ for arbitrary sets $F$ of at most $f$ edges.
Observe that in a directed graph, even if the graph $G\m F$ is still strongly connected,
the failure of $F$ may increase the replacement distance by an additive term of $n$ compared to 
the original distance $d(s,v,G)$.
We parameterize the problem by this increase.
Now the subgraph only has to preserve those distances that get stretched by at most $k$.
More precisely, we define an $(f,k)$-EFT-BFS for a given graph $G$
to be a subgraph $H$  such that, for every vertex $v \in V$ and every subset $F\subseteq E$ of size $|F| \le f$
satisfying $d(s,v,G\m F)\leq d(s,v,G)+k$, 
we have $d(s,v, H\m F) = d(s,v, G\m F)$.

In Section \ref{sec:FT-preservers} we prove the following results. 
The first one is a more refined bound on the size of $(f,k)$-\normalfont{EFT-BFS} parameterized by $k$.

\begin{theorem}
\label{thm:intro_FT-BFS}
For any parameters $f,k\geq 1$ and directed $n$-vertex, $m$-edge graph $G$, 
we can compute in $2^{O(fk)}mn$ time an $(f,k)$-\normalfont{EFT-BFS} with $2^{O(fk)}n$ edges.
\end{theorem}

Note that for parameters $f$ and $k$ with $f \cdot k = o(\log n/\log \log n)$ we obtain a truly subquadratic number of edges.
Furthermore, the result is optimal for constant $f$ and $k$ since $\Omega(n)$
is an existential lower bound on the size of $({s}\times V)$-distance preserver even in the absence of any failures. Finally, our result improves significantly over the $\widetilde{O}(f \nwspace n^{2-\frac{1}{2^f}})$ bound in the unparameterized case by Bodwin et al.~\cite{BGPW17}.

We conclude Section \ref{sec:FT-preservers} by presenting an oracle that reports for a given pair $(v,F)$ whether the failure of $F$
stretches the $s$-$v$-distance by at most an additive term $k$.
The preprocessing time of our oracle is independent of the number of supported failures $f$,
and the query time is independent of the size of the underlying graph $G$
Let $\omega<2.37286$ denote the matrix multiplication exponent~\cite{AlmanVWilliams21RefinedLaserMethod}.
Our result reads as follows.

\begin{restatable}{theorem}{oraclestretcheddistances}
\label{thm:intro_oracle_stretched_distances}
For any parameters $f,k\geq 1$ and directed graph $G = (V,E)$ with designated source $s \in V$, there is a 
Monte Carlo oracle
of size $O(2^k k^2 \cdot n^2)$
that,
for any pair $(v,F)$ where $v\in V$ and $F\subseteq E$ has size at most $f$, reports w.h.p. whether $d(s,t,G\m F) \le d(s,t,G) + k$ in time $O(2^k (fk)^\omega)$.
The oracle is computable in time $O(2^k k^\omega \cdot n^\omega)$, independent of $f$.
\end{restatable}

\section{Sensitivity Oracles for the $k$-Path Problem}
\label{sec:k-path}
In this section, we consider the \ftkpath problem
and design our solutions in the form of $k$-Path Sensitivity Oracles.
The first one aims for a query time that is polynomial in the parameter $k$
as well as the sensitivity $f$ and independent of the size of $G$,
the space needed to store the required information is rather high, roughly $k^{f+1}$.
The second sensitivity oracle reduces this space significantly, but the time needed to answer queries is now exponential in both $f$ and $k$, while still only logarithmic in the graph size.
Both data structures can be constructed in time proportional to their size and the time needed to detect a \kpath.

\subsection{Fault-Tolerant Lookup Trees}
\label{subsec:lookup_trees}

Our first solution is based on precomputing \kpaths in various subgraphs of $G$.
and arranging them in a structure that we call a \emph{fault tolerant-lookup tree} $FT^{f,k}(G)$.
The tree enables us to prepare for all possible failure sets 
since we can efficiently find a solution avoiding a given set $F$.
Let $T(k,n,m)$ be the time to compute a \kpath in an $n$-vertex, $m$-edge directed graph (or certify that no \kpath exits).
We tacitly assume $T(k,n,m) = \Omega(k \log k)$.

\begin{lemma}
\label{lem:tree}
	There is a $k$-Path Sensitivity Oracle of size $O(k^{f+1})$ 
    that can be computed in time $O(k^f \cdot T(k,n,m))$ and answers queries in time $O(f^2)$.
\end{lemma}

\begin{proof}
Given a graph $G = (V,E)$ and parameters $f$ and $k$, we compute a fault tolerant-lookup tree $FT^{f,k}(G)$.
\autoref{fig:kpath_tree} provides an illustrating example.

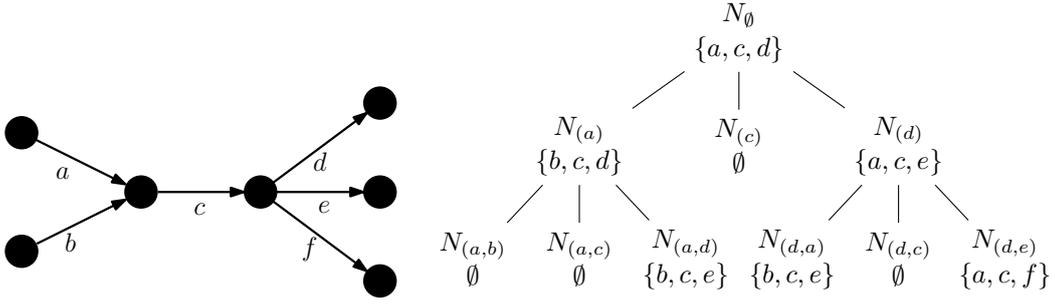
\begin{figure}
    \includegraphics[height=.2\textwidth]{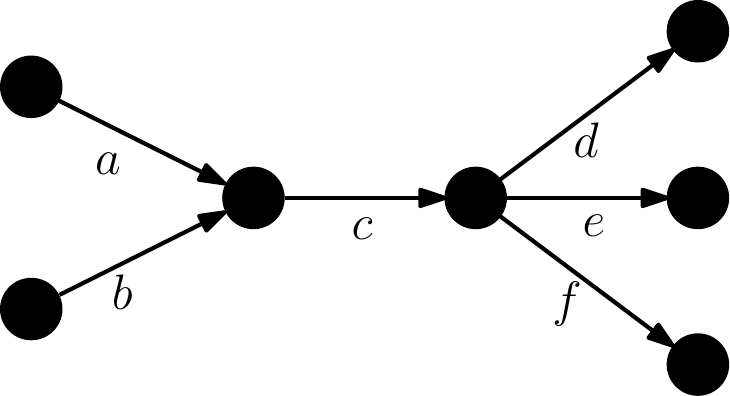}
    \hfill
    \begin{tikzpicture}[
    every node/.style = {shape=rectangle, align=center}, 
    level 1/.style={sibling distance={.15\textwidth}},
    level 2/.style={sibling distance={.1\textwidth}}]
    \node {\(N_\emptyset\)\\\(\{a,c,d\}\)}
        child { node {\(N_{(a)}\)\\\(\{b,c,d\}\)}
            child { node {\(N_{(a,b)}\)\\\(\emptyset\)}}
            child { node {\(N_{(a,c)}\)\\\(\emptyset\)}}
            child { node {\(N_{(a,d)}\)\\\(\{b,c,e\}\)}}}
        child { node {\(N_{(c)}\)\\\(\emptyset\)}}
        child { node {\(N_{(d)}\)\\\(\{a,c,e\}\)}
            child { node {\(N_{(d,a)}\)\\\(\{b,c,e\}\)}}
            child { node {\(N_{(d,c)}\)\\\(\emptyset\)}}
            child { node {\(N_{(d,e)}\)\\\(\{a,c,f\}\)}}};
  \end{tikzpicture}
    \caption{An example for an input graph $G$ (left) and a corresponding fault-tolerant lookup tree $FT^{f,k}(G)$ for $f=2$ and $k=3$ (right). Each node $N_{\vec S}$ is depicted with its stored \kpath $P(N_{\vec S})$.}
    \label{fig:kpath_tree}
\end{figure}

Let $E^{\le f}$ be the set of all vectors of length at most $f$ with entries from $E$.
For some $\vec{S} \in E^{\le f}$ and $e \in E$, $\vec S \diamond e$ is the vector obtained by appending $e$ to $\vec S$.
Finally, we let $S \subseteq E$ denote the set of entries of $\vec S$.
Each node in the tree $FT^{f,k}(G)$ is labeled as $N_{\vec{S}}$ for some $\vec{S}\in E^{\le f}$. 
Intuitively, the vector $\vec S$ corresponds to a path in $FT^{f,k}(G)$ starting at the root,
and the set $S$ corresponds to the failures considered in $N_{\vec{S}}$.
Each node stores a \kpath as a set of edges, denoted by $P(N_{\vec{S}})$.
We ensure that $P(N_{\vec S})$ and $S$ are always disjoint. 
We use $P(N_{\vec{S}})=\emptyset$ to mark the fact that there is no \kpath in the graph $G\m{}S$.

We set $N_\emptyset$ to be the root and $P(N_\emptyset)$ to any \kpath in $G$,
or $P(N_\emptyset) = \emptyset$ if no such path exists.
We iteratively augment this to a tree of depth up to $f$.
As long as there is a node $N_{\vec{S}}$ with $P(N_{\vec{S}}) \neq \emptyset$ and $|S|<f$,
we create a child node  $N_{\vec{S}\diamond e}$ of $N_{\vec{S}}$  for every $e\in P(N_{\vec{S}})$. We then decide whether there is a \kpath in $G-(S\cup\{e\})$. 
If so, we set it as the \kpath of $N_{\vec{S}\diamond e}$. If no such \kpath exists, we set $P(N_{\vec{S}\diamond e})=\emptyset$.
Note that indeed we have $P(N_{\vec{S}}) \cap S = \emptyset$ for every node $N_{\vec{S}}$. This implies that for every child $N_{\vec{S'}}$ of $N_{\vec{S}}$ we have $|S'|=|S|+1$ and hence after $\ell$ steps from the root $N_{\vec{S}}$ we have $|S|=\ell$.

We query the data structure with a set $F \subseteq E$ of failing edges.
We start at the root $N_\emptyset$ and traverse the tree as follows.
Let $N_{\vec{S}}$ be the current node. 
First assume $P(N_{\vec{S}})\neq \emptyset$. 
If $P(N_{\vec{S}}) \cap F = \emptyset$, we report that there is a \kpath in $G\m{}F$
since $P(N_{\vec{S}})$ is a witness of this fact. 
Otherwise, $P(N_{\vec{S}})\cap F \neq \emptyset$
and we move to $N_{\vec{S}\diamond e}$ for an arbitrary $e\in P(N_{\vec{S}})\cap F$.
By using only such steps, 
we guarantee that for every visited node $N_{\vec{S}}$ we have $S\subseteq F$. 
Hence, if $P(N_{\vec{S}}) = \emptyset$, there is no \kpath in $G-S$
and thus none in $G-F$ and we make the according output. 
Note that if $G-S$ contains a \kpath and $S' \subseteq S$ is a subset,
then also $G-S'$ contains a \kpath.
Therefore, it indeed does not matter which edge $e$ to follow during the query.
If after $|F|$ steps, we are at node $N_{\vec{S}}$, then $|S|=|F|$ whence $S=F$.
By construction, we have $P(N_{\vec{F}})=\emptyset$ or $P(N_{\vec{F}}) \cap F  = \emptyset$, so every query visits at most $|F| \le f$ nodes.  

The tree is of depth at most $f$ and every node has at most $k$ children, so the total number of nodes is bounded by $O(k^f)$.
To compute $P(N_{\vec{S}})\cap F$ efficiently, we store the \kpath in any node $N_{\vec{S}}$ in a static dictionary as proposed by Hagerup, Miltersen, and Pagh~\cite{Hagerup_Miltersen_Pagh_2001}. 
Their construction allows answering queries of the form $e\in P(N_{\vec{S}})$ in constant time, takes space in $O(k)$ and is computed in time $O(k \log k)$.\footnote{%
	The weak non-uniformity mentioned in \cite{Hagerup_Miltersen_Pagh_2001}, 
	that is, the need of compile-time constants depending on the word size, 
	only holds if this size is $\omega(\log n)$,
	which is not the case for us.
} 
Hence, the data structure requires space in $O(k^{f+1})$.  As for each of the up to $O(k^f)$ nodes we may have to find a \kpath and store the respective dictionary, the construction time is in $O(k^f\cdot (T(k,n,m)+k \log k)) = O(k^f\cdot T(k,n,m))$.
Finally, the query algorithm visit at most $f$ nodes and for every node $N_{\vec{S}}$ have to check whether there is $P(N_{\vec{S}}) \cap F \neq \emptyset$.
Using the respective dictionary, this can be done in time $O(|F|)=O(f)$.
The total query time is in $O(f^2)$.
\end{proof}

The current fastest deterministic algorithm for $k$-\textsc{Path} was given by Tsur~\cite{Tsur_2019}
and implies a preprocessing time of $O(k^f \cdot 2.554^k \cdot \poly(n))$ for our
oracle for \ftkpath.
Using the randomized algorithm by Williams~\cite{Williams_2009},
this reduces to $O(k^f \cdot 2^k \cdot \poly(n))$.
For more efficient algorithms on specific graph classes, we refer to the work by Uehara and Uno \cite{Uehara_Uno_2007}.

If $f>k$, we are able to improve both the query and preprocessing time. We use the same data structure as above but instead of building a dictionary for every \kpath we simply store the \kpaths. 
Conversely, at the beginning of each query, we build a dictionary for the failure set $F$ instead and then check every element of the visited \kpath in this dictionary, yielding the following result.

\begin{corollary}
   There is a $k$-Path Sensitivity Oracle of size $O(k^{f+1})$ that can be computed in time $O(k^f\cdot T(k,n,m))$ and answers queries in time $O(f \log(f) + fk)$.
\end{corollary}

\subsection{Subgraph Sampling}
\label{subsec:sampling}

Our second solution samples subgraphs of $G$ by randomly failing edges and computes \kpaths on the resulting instances. The query algorithm searches for a \kpath among the precomputed solutions 
that does not intersect the given set of failing edges.
If the existence of a fault-tolerant \kpath is reported, this is always correct.
However, due to the random preprocessing, the query algorithm may produce false negatives. 
We show that the algorithm succeeds on all possible queries with high probability.\footnote{
	We say an event occurs \emph{with high probability} (w.h.p.)
	if there is a constant $c>0$ such that the probability is at least $1-{1}/{n^c}$.
	In fact, in all our results the constant $c$ can be made arbitrarily large
	without affecting the asymptotic statements.
}
We define $r_{f,k} = \left(\frac{f+k}{f}\right)^{f} \left(\frac{f+k}{k}\right)^{k} \cdot 6\,f\log n$.

\begin{lemma}\label{lem:randkpath}
    There is a randomized $k$-Path Sensitivity Oracle with one-sided error 
    that takes space $O(r_{f,k}\,{\cdot}\,k)$, can be computed in time $O(r_{f,k}\,{\cdot}\,T(k,n,m))$,
    and in time $O(r_{f,k}\,{\cdot}\,f)$ answers queries correctly w.h.p.
\end{lemma}

\begin{proof}
    Given a graph $G=(V,E)$ and parameters $f$ and $k$, the precomputing algorithm creates $r_{f,k}$ spanning subgraphs of $G$ independently at random. Every spanning subgraph $H$ is created by removing any edge $e \in E$ with probability $p={f}/{(f+k)}$, independently of all other choices. 
Let $\mathcal{H}$ be the set of these subgraphs $H$. We run the \kpath algorithm on every $H\in\mathcal{H}$ and, if it returns a \kpath $K$, we add it to the set of found \kpaths $\mathcal{K}$. This takes time in $O(r_{f,k}\cdot T(k,n,m))$.
    
    Given a query $F\subseteq E, |F|\le f$, we decide whether there is $K\in\mathcal{K}$ with $K\cap F=\emptyset$. If so, we output $K$. Otherwise, we output \emph{no}. In order to efficiently compute overlaps with $F$, we employ static dictionaries \cite{Hagerup_Miltersen_Pagh_2001}. Hence, the total construction time is in $O(r_{f,k}\cdot (T(k,n,m)))$. We only have to store the dictionaries of $\mathcal{K}$, so the size of the data structure is in $O(r_{f,k}\cdot k)$. Every query has to potentially check every $K\in\mathcal{K}$ for an overlap with $F$. Having prepared the dictionaries, this takes time in $O(r_{f,k}\cdot f)$.
    
    It remains to show that w.h.p.\ all queries are answered correctly and the error is one-sided. Any found \kpath $K$ in some subgraph $H\in \mathcal{H}$ not intersecting $F$ is also a \kpath in $G-F$. We obtain that the algorithm only outputs feasible \kpaths to the queries. However, the precomputed set of \kpaths $\mathcal{K}$ might not contain a \kpath, which does not intersect $F$, even though $G-F$ contains one. Therefore, the error is one-sided.
    
    We call a precomputed subgraph $H \in \mathcal{H}$ a \emph{witness} for a given set $F \subseteq E$, if the \kpath algorithm returns a \kpath K with $K \cap F = \emptyset$. Note that the query cannot produce a false negative for a given set $F$ if $\mathcal{H}$ contains a witness for $F$. We know that $H$ is a witness for $F$ if $H$ was sampled by removing at least the edges in $F$ and keeping at least one fixed \kpath.
    \[
        \prob{H \text{ witness}} \ge p^f \cdot (1-p)^k.
    \]
    We denote the latter by $q_{f,k}$, which, by standard analysis, is maximal for the chosen $p=f/(f+k)$. Note that $r_{f,k} = q_{f,k}^{-1}\cdot 6f\log n$. By a Chernoff bound~(see for example~\cite{Mitzenmacher_Upfal_2017}), we conclude that sampling $r_{f,k}$ graphs suffices to have w.h.p.\ at least one witness for every $F\subseteq E$ with $ |F|\le f$. For a given $F$, let $X$ be the set of witnesses for $F$ in $\mathcal{H}$. Since the subgraphs in $\mathcal{H}$ are sampled independently at random, we obtain the following with $\delta=1$.
    \[
       \prob{|X| = 0} = \prob{|X| \le (1 -\delta) r_{f,k} \ q_{f,k}} \le 
       	\exp\!\left(-\frac{\delta^2}{2} r_{f,k} \ q_{f,k} \right) = 
       	\exp\!\left(-\frac{1}{2}\cdot 6f\log n\right) = \frac{1}{n^{3f}}.
    \]
    By applying the union bound we get that the probability of not having a witness for any of the possible $m^f\le n^{2f}$ is upper bounded by $n^{-f}$.  
    As $f>0$, we get that the data structure answers all queries correctly w.h.p.
\end{proof}

Just like for \autoref{lem:tree}, we are able to improve the data structure for $f>k$ by creating a dictionary of $F$ instead of dictionaries for the \kpaths.
We can safely assume $f \le m = \poly(n)$, 
whence the extra $O(f \log f)$ time spend to build the dictionary
is negligible compared to the remaining query time of $O(r_{f,k}\cdot k) = O\big(\big(\frac{f+k}{f}\big)^{f} \big(\frac{f+k}{k}\big)^{k} \cdot fk \log n\big)$.

\begin{corollary}
    There is a randomized $k$-Path Sensitivity Oracle with one-sided error 
    that takes space $O(r_{f,k}\cdot k)$, can be computed in time $O(r_{f,k} \cdot T(k,n,m))$,
    and in time $O(r_{f,k}\cdot k)$ answers queries correctly w.h.p.
\end{corollary}

The term $r_{f,k}$ is a bit unwieldy, so we give an intuition in the form of an upper bound. The  number of subgraphs $r_{f,k}$ grows as the probability $q_{f,k}$ decreases. This is minimal for $f=k$ with $p=\frac{1}{2}$. Hence, $r_{f,k}$ is maximal for $f=k$.

\begin{corollary}
    There is a randomized $k$-Path Sensitivity Oracle with one-sided error 
   that takes space $O(2^{f+k} \cdot fk \log n)$, can be computed in time $O(2^{f+k} \cdot T(k,n,m) \cdot f\log n)$, and in time $O(2^{f+k} \cdot f \min\{f,k\} \log n)$ answers queries correctly w.h.p.
\end{corollary}

\subsection{Related Problems}

Indeed, it is straight forward to adapt our techniques to all types of problems that search for a subgraph containing $k$ edges. For such a problem $\Pi$ \autoref{lem:tree} and \autoref{lem:randkpath} directly transfer when $T(k,n,m)$ is replaced by the best running time to exactly solve $\Pi$ with parameter $k$ on an $n$-vertex, $m$-edge graph. Such an approach, for instance, can be applied to \textsc{$k$-$(s,t)$-Path} (the variation of \textsc{$k$-Path} with specified start and end vertex $s$ and $t$), \textsc{Graph Motif} (for a definition, see e.g.~\cite{Pinter2014DeterministicPA}), \textsc{$k$-Tree} and \textsc{$k$-Subgraph Isomorphism} for graphs of bounded treewidth (for a definition, see e.g.~\cite{Fomin2014EfficientCO}). 

There are several other problems that can be solved by a reduction to \textsc{$k$-Path}, although they do not fall into the above category. In the fault-tolerant setting, it is not trivial to reduce to a \ftkpath in a useful manner. 
One example for such a related problem is  \textsc{Exact Detour}~\cite{BezakovaCDF19Journal}, where given a graph $G=(V,E)$, two vertices $s, t\in V$ and a parameter $k\in \mathbb N$, the task is to decide if there exists a path of length $d(s,t,G)+k$ in $G$.
We can exploit a $k$-Path Sensitivity Oracle plugging it into the reduction from $k$-\textsc{Path} given in~\cite[Lemma 4.5]{BezakovaCDF19Journal} to obtain a data structure for a fault-tolerant version of  \textsc{Exact Detour}. 
Specifically, we consider an $f$-edge fault-tolerant SO that builds on input $(G,s, t,k)$ and $f$, where $G=(V,E)$ is a graph, $s,t\in V$ and $k,f\geq 0$, a data-structure that can decide, given any set of edge failures $F\subseteq E$ with $|F|\leq f$, if there exists a path of length $d(s,t,G)+k$ in  $G-F$. Note that this is not the same as solving  \textsc{Exact Detour} for $(G\m F,s,t,k)$ since we consider a detour with respect to the distance in the original graph $G$ without failures.

\begin{lemma}
There exists an $f$-edge fault-tolerant sensitivity oracle for the \textsc{$k$-Exact Detour} problem of size $O(n^2\cdot k^{f+2})$ that can be built in time $O(n^2\cdot k^{f+1} \cdot T(k,n,m))$ and answers for every $F\subseteq E$ with $|F|\leq f$ queries in time  $O(n^2\cdot (k f^2+k^2))$, where $T(k,n,m)$ denotes the time to solve \textsc{$k$-$(s,t)$-Path} on an $n$-vertex, $m$-edge graph.
\end{lemma}
\begin{proof}
To design the oracle we first prepare an $f$-edge fault-tolerant oracle for the \textsc{$k$-$(s,t)$-Path} requests in the reduction from~\cite[Lemma 4.5]{BezakovaCDF19Journal}. 
First partition the vertices of $G$ into the layers of the BFS tree starting at $s$. Then we use \autoref{lem:tree}  to build, for every pair of vertices $(x,y)$ and for every $1\leq \ell’\leq 2k+1$, an 
$f$-edge fault-tolerant lookup tree for $\ell'$-paths
denoted by $T_{x,y,\ell'}$ on the graph induced by the vertices in the levels between $x$ and $y$. 

On query $F$, we build the range-$k$ detour graph $(G-F)^{(k)}$ for $G-F$ like in~\cite{BezakovaCDF19Journal} with one alteration. Instead of using an oracle for \textsc{Exact Path} to check for tuple $(x,y,\ell')$ if there exists a path of length $\ell'$ between $x$ and $y$ in $G$, we use the look-up tree $T_{x,y,\ell'}$ to query $F$ to decide if there exists a path of length $\ell'$ between $x$ and $y$ in $G-F$.

Similar to ~\cite{BezakovaCDF19Journal}, the graph $(G-F)^{(k)}$  can be used to decide if there is a simple path from $s$ to $t$ in $G-F$ of length $k$. 
We describe a variant of it --- consider a DAG (which is a multi-graph as there might be several edges of different weights between a pair of vertices), such that for every pair of vertices $x,y \in V$ and for every integer $0 \le \ell' \le 2k+1$ there is an edge from $x$ to $y$ of weight $\ell'$ in the DAG iff there is a simple path from $x$ to $y$ of length $\ell'$ in $G-F$. It is not difficult to verify that for every vertex $t \in V$ there is a simple path from $s$ to $t$ in $G-F$ of length $d(s,t,G)+k$ iff there is a path in the DAG from $s$ to $t$ such that the sum of its edge weights is $d(s,t,G)+k$. 
Following the dynamic programming described in \cite{BezakovaCDF19Journal} (Section 4), one can check in additional $O(n^2 k^2)$ time if there is a simple path of length $d(s,t,G)+k$ in this DAG.

By \autoref{lem:tree}, we can construct all the lookup trees in $O(n^2\cdot k^{f+1} \cdot T(k,n,m))$ time, and the overall structure we build has a size in $O(n^2\cdot k^{f+2})$. By~\cite{BezakovaCDF19Journal}, constructing $(G-F)^{(k)}$  requires $O(n^2\cdot k)$ calls to a \textsc{$k$-Path} oracle.
Hence the overall query time is in  $O(n^2\cdot (k f^2+k^2))$.
The correctness follows as in~\cite{BezakovaCDF19Journal}.


\end{proof}

For  \textsc{Long-$(s,t)$-Path} or \textsc{Long Cycle}~\cite{Fomin2018LongD}, the case of a yes instance with a path (cycle) exceeding $2k$ appears to be even trickier to adapt. Perhaps combining an oracle for $k$-path with a reachability oracle can be used to build structures for fault tolerant versions of these problems.

\section{Sensitivity Oracles for the Vertex Cover Problem}
\label{sec:vertex_cover}

We treat here the design of $f$-edge fault-tolerant oracles for the \textsc{$k$-Vertex Cover} problem. 
First of all, note that in this setting if there exists a vertex cover of cardinality $k$ in $G$ (without any failures), then  $(G,k,F)$ is a yes-instance for each $F\subseteq E$ with $|F|\leq f$.
On the other hand, if there exists no vertex cover of size $k+f$ then $(G,k,F)$ is a no-instance for each $F\subseteq E$ with $|F|\leq f$.

It is tempting to think that it is possible to compute a kernel with $2(k+f)$ vertices that contains a $k$-vertex-cover for every $G-F$ that admits it, where $F$ is a set of at most $f$ edges of $G$, by the linear kernel for the vertex cover problem. However, the linear programming approach does not (directly) transfer to the fault-tolerant setting. Adding failure-variables to the linear program destroys the half-integrality of the program. Building a $2(k{+}f)$-kernel without adjustments may remove crucial vertices for the fault tolerant solution. Consider as a small example a $C_4$ with one leaf attached to three of the vertices for  $f=k=2$. The LP-kernel only keeps the degree-3-vertices, however there are 2-failure sets that then allow for a vertex cover of size 2, but only when using the removed degree-2-vertex in the solution.
A closer look reveals that in general there is no kernel with less than $kf$ vertices, as can be seen by the simple example of $k+1$ stars with $f$ leafs each, where clearly each edge is needed to identify the sets $F$ for which $G\m F$ contains a vertex cover of size $k$. 

The high-degree reduction rule (usually called \emph{Buss}-rule, and attributed to \cite{bussrule}) from the classical $k^2$ vertex cover kernel however can easily be adapted as follows.

\begin{lemma}\label{lem::vckernel}
For any integers $f,k\geq 1$ and any graph $G$, we can compute in polynomial time a subgraph $H$ of $G$  with $O(fk + k^2)$ vertices and edges and a bound $k'\leq k$ such that for any $F\subseteq E$ with $|F|\leq f$, $G-F$ has a vertex cover of size $k$ if and only if $H\m F$ has a vertex cover of size $k'$.
\end{lemma}

\begin{proof}
Given a graph $G$ and $f,k\geq 1$ apply the following reduction exhaustively. If $G$ contains a vertex $v$ of degree more than $k+f$, remove $v$ and its adjacent edges from $G$ and decrease $k$ by one.  Further, remove vertices of degree zero.

To see that this reduction is safe, let $F\subseteq E$ be a set of edges with $|F|\leq f$ for which $G\m F$ contains a vertex cover $C$ of cardinality $k$. Any vertex $v$ of degree more than $f+k$ in $G$ has degree more than $k$ in $G\m F$, which implies that $v$ has to be included $C$. 

After exhaustive application of this rule, there cannot be more than $f+k(f+k)$ edges remaining in the graph (observe that all edges have to be covered by a vertex or be deleted), otherwise $(G,k,F)$ is a no-instance for each $F\subseteq E$ with $|F|\leq f$. Since there are no vertices of degree zero, there hence are at most $O(f+k(f+k)) = O(fk + k^2)$ vertices remaining. 
\end{proof}

This reduction can be used as starting point for every Vertex Cover Sensitivity Oracle to remove dependence on $n$. Together with an exact algorithm to solve vertex cover on $H\m F$, this yields a Vertex Cover Sensitivity Oracle of size $O(fk+k^2)$ with query time dependent only on $f$ and $k$, for example with query time $O(1.2738^k+fk^2)$ when using the current best parameterized algorithm for vertex cover by Chen, Kanj, and Xia~\cite{bestvcfpt}. 

High-degree reduction is also the basis for the dynamic data structures for  \textsc{Vertex Cover} in~\cite{dynamicfpt,dynamicfpt2}. A slight adaption of the structure by Alman, Mnich and Vassilevska Williams~\cite{dynamicfpt} gives an oracle of the same size with an improved query time in $O(1.2738^k+f)$. More precisely, their original construction has size $O(nk+k^2)$ and single-edge worst-case update time $O(k)$, which translates in our fault-tolerant setting to a query time in $O(1.2738^k+fk)$ when combined with the current best parameterized algorithm for vertex cover.  Knowing that the structure only has to be robust against a one-time event of $f$ edge failures, it suffices to store only up to $f$ edges in the lists $R_v$ which yields a structure of size in  $O(fk + k^2)$. Further, in the update procedure for deleting an edge $(u,v)$, the expensive part in case an edge with $deg(u)=k$ is removed  is not necessary which reduces the worst-case update time to $O(1)$. These observations directly yield the following.  

\corintroVCSOlowspace*

As we have seen for the \textsc{$k$-Path} problem and also in the work of Alman, Mnich and Vassilevska Williams~\cite{dynamicfpt}, branching-tree structures can be useful to guarantee very fast query times.
A first natural approach would be to adapt the standard $2^k$-branching on edges for vertex cover, and add to it a third branch that considers edge-failure. While this branching tree nicely finds vertex covers of size at most $k$ for up to $f$ failures, it is not clear how to efficiently answer a specific query faster than looking at the whole search tree which has size $O(3^{f+k})$. In the following, we give two variations of this idea, each with an improved query time depending on the relation between the values $f$ and $k$. 
The first idea gives a faster query time for the case that $f$ is small in comparison to $k$.

\thmintroVCSOthreetof*
\begin{proof}
Given a graph $G=(V,E)$ and $f,k\geq 0$, by  \cref{lem::vckernel} we can assume, at the cost of $O(poly(n))$ preprocessing, that $G$ has $O(k^2+kf)$ vertices. We use branching to compute a set of solution pairs. Starting with the input graph $G$ and empty sets $C$ and $F$, we branch as follows each time creating a subgraph of $G$. If the current subgraph $G'$ still contains an edge $(u,v)\in E(G')$, branch into the cases \begin{enumerate}
\item If $|C|<k$, add $u$ to $C$, and delete $u$ from $V(G')$ along with edges incident to $u$ from $E(G')$
\item If $|C|<k$, add $v$ to $C$, and delete $v$ from $V(G')$ along with edges incident to $v$ from $E(G')$
\item If $|F|<f$, add $(u,v)$ to $F$ and delete  $(u,v)$ from $E(G')$.
\end{enumerate}
Let $(F,C)$ be a leaf in the created branching tree with associated subgraph $G'$, which we emphasize to be created dynamically. $G'$ does not contain edges either in $F$ or incident to vertices in $C$. If $G'$ does not contain an edge at all, we add $F$ to the set $\mathcal{F}$. Note that in this case $C$ is a vertex cover for $G\m F$ with $|F|\leq f$ and $|C|\leq k$.
We do so for all leafs $(F,C)$ in the branching tree, which by the termination constraints $|C|<k$ and $|F|<f$ contains at most $3^{k+f}$ many leafs.

To build an \ftkvc, we create a hashing table for all the collected sets $\mathcal{F}$. On query $F$, we look up for each $F'\subseteq F$. If $F'$ is among the collected sets in $\mathcal{F}$, and  answer yes if and only if we find such a set. This yields the claimed query time $O(2^f)$.
To see why the answer we give is correct, first observe that if we find an $F'\subseteq F$ in the collected sets, then there exists a vertex cover of cardinalty at most $k$ in $G\m F'$ and hence also in $G\m F$ so we correctly answer yes. Conversely, if there exists a vertex cover $C$ of cardinalty at most $k$ in $G\m F$, then our branching procedure has created at least one $F'\subseteq F$: Consider the recursive branching procedure, and follow always the branch that either adds a vertex in $C$ to the cover or an edge in $F$ to the failing set (note that since $C$ is a vertex cover in $G\m F$, there is always at least one branch that meets one of these criteria). This path through the branching tree has to at some point create a graph where no edges remain, at the latest when $C$ and $F$ have been collected completely. If it terminates before, then $G$ is empty for some $C'\subseteq C$ and $F'\subseteq F$, and thus $F'$ is among the collected sets.
\end{proof}

In a sense, \autoref{thm:intro_VCSO_3_to_f} adds the failure branch to the standard $2^k$ branching for \textsc{$k$-Vertex Cover}. It is an interesting challenge to augment more sophisticated branching algorithms with failure branches, such as the one by Chen, Kanj, and Xia~\cite{bestvcfpt}.

For now, we did not actually use branching to build a lookup tree and exploit the insights from the branching decisions.
With a measure \& conquer approach, we can do a different type of branching that builds a useful lookup tree, with a larger size but very fast query time.

\thmintroVCSOlowquery*
\begin{proof}
Given a graph $G=(V,E)$ and $f,k\geq 0$, by  \cref{lem::vckernel} we can assume, at the cost of $O(poly(n))$ preprocessing, that $G$ has $O(k^2+kf)$ vertices. We compute a binary fault-tolerant lookup tree $FT^{f,k}(G)$ with nodes labeled $N_{S,F}$ for some sets $S,F\subseteq E$. We build this tree by branching on edges and the sets $S$ and $F$ in the subscript of the nodes represent the edges for which we decided during the branching to be either failing (in $F$) or not failing ($S$ for \emph{safe}). The underlying idea to bound the number of such nodes we have to create is that we can budget the safe edges by $\frac{1}{k+1}$ as follows. As soon as we have decided on $k+1$ edges adjacent to some vertex $v$ to be safe, we know that any vertex cover of cardinality at most $k$ for the graph induced by the edges in $S$ has to contain $v$. Thus we remove $v$ from $G$ before further branching to ensure that for each vertex at most $k+1$ adjacent edges are placed in $S$. With this sort of reduction rule during the branching, we know that we can stop branching on further edges being in $S$ as soon as $|S|=(k+1)k$.

In order to keep track of the edge-decisions, we store in each node $N_{S,F}$ not only the branching edge $e_{S,F}$ but also a subgraph $G_{S,F}$ of $G$. The latter allows us to keep track of vertices that were removed during the branching. We start by initializing the root $N_{\emptyset,\emptyset}$ storing $G_{\emptyset,\emptyset}=G$ and $e_{\emptyset,\emptyset}=\emptyset$. For each node $N_{S,F}$  with $e_{S,F}=\emptyset$ and $G_{S,F}\neq \emptyset$ do the following.

If there still exists an edge  in $G_{S,F}$ that is not in $S$, pick one such edge $(u,v)$, set  $e_{S,F}=(u,v)$, and attach to $N_{S,F}$ the following (at most) two children 
\begin{enumerate}
\item If $|F|<f$, attach with an arc labelled 0 a child node labeled $N_{S,F'}$ where $F'=F\cup(u,v)$. Set $G_{S,F'}$ to the graph obtained from  $G_{S,F}$ by deleting $(u,v)$ and set $e_{S,F'}=\emptyset$.
\item If $|S|<k(k+1)$, consider $S'=S\cup\{(u,v)\}$ and $G'=G_{S,F}$.  If there exists a vertex in  $G'$, that is adjacent to more than $k$ edges in $S'$, remove it from $G'$ along with incident edges.
After exhaustive application of this reduction, attach with an arc labeled 1 a child node labeled $N_{S',F}$ with $G_{S',F}=G'$ and $e_{S',F}=\emptyset$.
\end{enumerate}
If there exists no edge in $G_{S,F}$ that is not contained in $S$, check if there exists a vertex cover of cardinality at most $k$ in $G-F$.In case there exists such a cover, set $G_{S,F}=\emptyset$.

By the termination constraints $|F|<f$ and $|S|<k(k+1)$, we see that this tree contains at most $2^{f+k(k+1)}$ leafs. Further, the computation cost is bounded by  $O(2^{f+k^2+k}(1.2738^k +\poly(f,k))+\poly(n))$, where the $poly(f,k)$-term includes the needed vertex/edge-deletions, and we use the current-best \textsc{Vertex Cover} algorithm by Chen, Kanj, and Xia~\cite{bestvcfpt} running in time $O(1.2738^k + kn)$ for a graph on $n$ vertices.

On query $F$, we use this lookup tree $FT^{f,k}(G)$ as follows. We first build a dictionary for $F$. Then, starting at the root, we move from a node $N_{S,F'}$ to one of its children via the arc labeled 0 if $e_{S,F'}\in F$ and via the arc labeled 1 otherwise. We then answer yes if and only if this path in the lookup tree terminates in a node $N_{S,F'}$ with $G_{S,F'}=\emptyset$. Since the depth of $FT^{f,k}(G)$ is bounded by $f+k(k+1)$, this gives a query time in $O(f+k^2)$.

To show that this procedure gives the correct answer, first assume that on input $F\subseteq E$ we answer with yes. This means that the path followed to answer the query in  $FT^{f,k}(G)$ terminated in some node $N_{S,F'}$ with $G_{S,F'}=\emptyset$. The way we followed a path to answer query $F$ ensured that $F'\subseteq F$. Now note that there are two possibilities that $G_{S,F'}$ is the empty graph. One possibility is that during the construction we found a vertex cover of cardinality at most $k$ in $G-F'$. In this case there obviously also exists a vertex cover of cardinality at most $k$ in $G-F$ and our answer is correct. The other possibility is that we removed all vertices with applications of the reduction in step 2. For every vertex that is removed in this step, we picked $k+1$ edges adjacent to it into $S$. Hence if this procedure manages to delete all vertices from $V$, then $(k+1)|V|\leq |S|$ which gives $|V|\leq k$ and there trivially exists a vertex cover of cardinality at most $k$ for every possible query $F$.

Conversely, assume that there exists a vertex cover $C$ of cardinality at most $k$ in $G\m F$. Assume towards contradiction, that on query $F$, the path in  $FT^{f,k}(G)$ terminates in a node  $N_{S,F'}$ with $G_{S,F'}\neq\emptyset$. By the route we follow for answering the query $F$, we know that $F'\subseteq F$ and that $S\cap F =\emptyset$. A leaf in the lookup tree with non-empty graph exists in two cases, either there was no edge to branch on and we did not find a suitable vertex cover (indicated by $e_{S,F'}=\emptyset$), or there was an edge to branch on but there was no suitable child attached because of the termination rules ($e_{S,F'}\neq \emptyset$).

If there was no edge to branch on,  $G_{S,F'}$ was not set to be empty, because there is no vertex cover of cardinality at most $k$ in $G-F'$. In particular, $C$ is not a vertex cover for $G\m F'$. Since it is a vertex cover for $G\m F$, it follows that the edges not covered by $C$ are included in $F\setminus F'$. 
Since there was no edge to branch on, in particular all edges in $F$ were deleted from $G_{S,F'}$ during the construction of the lookup tree. In this case, an edge in $F\setminus F'$ was deleted by step 2, because it was of the form $(u,v)$ such that more than $k$ edges in $S'\subseteq S$ are incident to $u$. Since $(S'\cap F)\subseteq (S\cap F)=\emptyset$, it follows that $C$ covers all edges in $S'$. To do so, $C$ has to contain the vertex $u$, as otherwise $C$ would have a cardinality of more than $k$. This means however that $C$ also covers all edges in $F\setminus F'$, a contradiction.

If there exists an edge $(u,v)$ in $G_{S,F'}$ that is not in $S$ but there was no suitable arc leading from the node $N_{S,F'}$ the according budget has to be expired. In particular, if $(u,v)\in F$, there is no arc labeled 0 from $N_{S,F'}$ because $|F'|=f$. Since however $F'\subseteq F$, this is not possible, without $F'=F$ and thus $(u,v)\notin F$. For $(u,v)\notin F$ there is no arc labeled 1 from $N_{S,F'}$ because $S$ contains $k(k+1)$ edges from $E\setminus F'$. The reduction step used in 2 makes sure that there exists no vertex in $G_{S,F'}$ that covers more than $(k+1)$ edges in $S$. Since $C$ in particular covers the edges in $S$, it has to contain $k$ vertices that each cover exactly $k+1$ edges in $S$. By the reduction rule, this means that all vertices of $C$ were deleted from  $G_{S,F'}$. Thus  $(u,v)$ is an edge in $G$ with $\{u,v\}\cap C=\emptyset$ and $(u,v)\notin F$, a contradiction to $C$ being a vertex cover in $G\m F$.

At last, note that for answering the queries, we do not actually need the whole graphs $G_{S,F}$ as information; we actually only need to distinguish between $G_{S,F}=\emptyset$ and $G_{S,F}\neq\emptyset$. After building the tree, we can replace every graph $G_{S,F}\neq\emptyset$ by a single vertex $\{v\}$, which reduces the overall space of the structure to $O(2^{f+k^2+k})$.
\end{proof}

\subparagraph*{Fault-Tolerant DSOs parameterized by the size of a vertex cover.} 
We conclude this section by showing an interesting connection between vertex cover and fault-tolerant DSOs. We recall that an $f$-edge/vertex fault-tolerant DSO is a data structure that reports the value $d(s,t,G-F)$ when queried on two vertices $s$ and $t$, and a set $F$ of edges/vertices of $G$ of size at most $f$. First of all, we observe that given a graph $G$ and a failing set $F$, if the graph $G\m F$ has a vertex cover of size bounded by $k$, the distance between any pair $(x,y)$ of vertices that are connected in $G-F$ is upper bounded by $2k$. This means that $d(x,y,G\m F) \leq d(x,y,G) + 2k$.

Further, we can use a vertex cover of $G$ as a building block for a DSO.

\begin{theorem}
For any $f\geq 1$, any $n$-vertex, $m$-edge graph $G=(V,E)$,
and $k$-vertex cover $C$ for $G$, 
we can compute in polynomial time an $f$-DSO of size $O(\min\{nk + fk^2, m\})$ 
with a query time of $\poly(k,f)$, that is, independently of the graph size.
\end{theorem}
\begin{proof}
Let $H$ be a graph whose vertex set is initially set to $C$ and whose edge set is initially empty.
We cycle through all pairs of vertices $x,y \in C$ with distance at most $2$ in $G$.
If $(x,y)$ is an edge in $G$, we add it to $H$ as well.
If the set $S = N_G^{\text{out}}(x) \cap N_G^{\text{in}}(y)$ of common neighbors has at most $f$ elements, we add them all to $H$ with the respective edges $(x,z)$ and $(z,y)$ for all $z \in S$.
Otherwise, that is if $|S| \ge f+1$, we instead add a new vertex $z'$ (not previously in $G$ or $H$)
and the edges $(x,z')$ and $(z',y)$.
In the end, $H$ has $O(fk^2)$ vertices and edges.
The oracle stores $H$ and all the edges between $v$ and $C$ for all vertices $v \in V{\setminus}C$,
whence the total size is at most $O(nk + fk^2)$.
To see why it is also in $O(m)$,
observe that the only time we introduce new edges not previously in $G$,
namely  $(x,z')$ and $(z',y)$, is if $|S| \ge f+1$.
In this case, we leave out the $2|S| > 2$ edges from $x$ to $S$ and from $S$ to $y$ in $G$.
 
On query $(u,v,F)$, we first compute the all-pairs distances in $H\m F$.
The anwer of the oracle depends on whether $u$ or $v$ are in the vertex cover $C$:
\begin{itemize}
	\item if $u,v\in C$, we simply return $d(u,v,H\m F)$;
	\item if $u\notin C$ and $v\in C$, we return $1+\min\{d(x,v, H\m F)\mid x\in N_{G-F}^{\text{out}}(u) \cap V(H)\}$;
	\item if $u\in C$ and $v\notin C$, we return $1+\min\{d(u,y,H\m F)\mid y\in N_{G-F}^{\text{in}}(v) \cap V(H)\}$;
	\item if $u,v\notin C$, 
		we return $2+\min\{d(x,y, H\m F)\mid x \in N_{G-F}^{\text{out}}(u) \cap V(H); 
			y \in N_{G-F}^{\text{in}}(v) \cap V(H)\}$.
\end{itemize}

\noindent
The correctness simply follows from the fact
that $f$ edge failures cannot disconnect two vertices $x$ and $y$ in $H$
if they have $|S| \ge f+1$ common neighbors.
Also, both the first and last edge of any shortest path between $u$ and $v$ in $G-F$ is covered by $C$.

Note that all operations on $H$ are polynomial in the size of $H$, 
and hence only depend on $k$ and $f$. 
Further, each vertex outside of $C$ only has neighbors in $C$ which means that checking their neighborhood and cross-checking if the corresponding edge is in $F$ only requires time in $O(k+f)$. The total query time is thus bounded of order $\poly(k,f)$.
\end{proof}




\section{Fault-Tolerant Distance Preservers of Bounded-Stretch}
\label{sec:FT-preservers}

Let $G=(V,E)$ be a directed graph, $s \in V$ be a source vertex, and $k$ a parameter. In this section we first consider the problem of computing an $(f,k)$-EFT-BFS of $G$, i.e., a sparse subgraph $H$ of $G$
such that, for every vertex $v \in V$ and every subset $F\subseteq E$ of size $|F| \le f$
for which $d(s,v,G\m F)\leq d(s,v,G)+k$ holds, 
we have $d(s,v, H\m F) = d(s,v, G\m F)$. Our results can be extended to handle also vertex failures, but the parameter $k$ doubles as for \textsc{$k$-Path} problem. We only present the results for edge failures.
We present an algorithm that computes a sparse $(f,k)$-\normalfont{EFT-BFS}  of $G$
with $2^{O(fk)}n$ edges.

Let $T$ be a shortest path tree of $G$ rooted at $s$.
Define $\level(v)$ to be depth of $v$ in $T$, that is, we have $\level(v) = d(s,v,G)$.
For a pair $(v,F)$, let $P_{v,F}$ denote (an arbitrary fixed) \emph{replacement path},
a shortest path from $s$ to $v$ in the graph $G\m F$. 
Let $e_1=(a_1,b_1), \ldots, e_L=(a_L,b_L)$ be the $L= d(s,v,G\m F)$ 
(not necessarily consecutive) 
edges of $P_{v,F}$.
We define $X(P_{v,F})$ to be the sequence of length $L$ whose $i^{th}$ entry is

$$X(P_{v,F})_i= 1+\level(a_i) - \level(b_i).$$

\noindent
Observe that for each edge $(a,b)\in P_{v,F}$, with $a$ appearing before $b$ in $P$, 
we have $\level(b)\leq \level(a)+1$. Thus, all entries of sequence $X(P_{v,F})$ are non-negative.
Let $X^+(P_{v,F})$ be restriction of $X(P_{v,F})$ to positive integers.
The following lemma presents an important property of sequence $X^+(P_{v,F})$.

\begin{lemma}
For each $v \in V$ and $F \subseteq E$ with $|F| \le f$,
the sum of entries in $X^+(P_{v,F})$ is equal to $d(s,v,G\m F)-d(s,v,G)$.
\label{lemma:X_property}
\end{lemma}

\begin{proof}
Let $e_1=(a_1,b_1), \ldots, e_L=(a_L,b_L)$ be the edges of $P_{v,F}$.
We can rearrange them such that $a_1$ is the source $s$ and $b_L = v$ the target
without affecting the sum over the entries of the sequence $x\in X^+(P_{v,F})$.
We get
\begin{align*}
\sum_{x\in X^+(P_{v,F})} x
	&=\sum_{x\in X(P_{v,F})} x =\sum_{i=1}^L(1+\level(a_i) - \level(b_i))\\
&=\sum_{i=1}^L1 +(\level(a_1) - \level(b_L)) = d(s,v,G\m F) - d(s,v,G). \qedhere
\end{align*}
\end{proof}

Let $\Y$ denote the set of all sequences of positive integers whose sum is bounded by $k$. 
We have $|\Y|=2^k-1$. 
Given a sequence $Y=(y_1,y_2,\ldots,y_t)\in \Y$, 
we define an auxiliary directed graph~$\G_Y$ on $(t{+}1) \nwspace n$ vertices as follows.

\begin{itemize}
\item The vertex set of $\G_Y$ is initialized to $t+1$ copies of set $V$, denoted by 
$V_{Y}^0,V_{Y}^1,\ldots,V_{Y}^t$.
The copy of vertex $v\in V$ in the set $V_{Y}^i$ is denoted by $v_{Y}^i$.\vspace*{.25em}
\item For $0 \le i \le t$ and $e=(a,b)\in E$,
we add a directed edge $(a_{Y}^i,b_{Y}^i)$ to $\G_Y$ if and only if $\level(b)-\level(a)=1$.\vspace*{.25em}
\item
For $1 \le i \le t$ and $e=(a,b)\in E$, 
we add directed edge $(a_{Y}^{i-1},b_{Y}^{i})$ to $\G_Y$ if and only if $1+\level(a)-\level(b)=y_i$.
\end{itemize}

\noindent
We now define the projection $\phi$ of entities in graph $\G_Y$ to the original graph $G$.
For a vertex $v^i_Y\in V(\G_Y)$, the projection $\phi(v^i_Y)$ is just the vertex $v$.
For an edge $(a^i_Y,b^j_Y)\in E(\G_Y)$, we define,  $\phi(a^i_Y,b^j_Y)$ to be the edge $(a,b)\in E(G)$.
For a subset $F$ of edges in $G$, we define $\phi^{-1}(F)$ to be the set of all those edges 
$(a^i_Y,b^j_Y)$ in $\G_\Y$ for which $\phi(a^i_Y,b^j_Y)$ lies in $F$.
Observe that, for any sequence $Y=(y_1,y_2,\ldots,y_t)\in \Y$, it holds that
$|\phi^{-1}(F)|\leq(t+1)|F|\leq (k+1)|F|$.

\begin{lemma}
Consider a triplet $(v,F,Y)$, where  $v\in V$, $F\subseteq E$ has size at most $f$, and $Y=(y_1,\ldots,y_t)\in\Y$.
There is an $s$-$v$-path $P$ in $G\m F$ satisfying $X^+(P)=Y$
if and only if
there is a path from $s_{Y}^0$ to $v_Y^t$ in $\G_Y\m \phi^{-1}(F)$.
\label{lemma:calH_Y}
\end{lemma}

\begin{proof}
First, suppose there is a path $P$ from $s$ to $v$ in $G\m F$ satisfying $X^+(P)=Y$.
Let $e_1=(a_1,b_1), \ldots, e_t=(a_t,b_t)$ be the edges in $P$ with $X^+(P_{v,F})_i = 1+\level(a_i) - \level(b_i)>0$.
Observe that the $e_i$ need not to be consecutive in $P$.
By assumption, for any $1 \le i \le t$,
the edge $(a_i,b_i)$ satisfies $1+\level(a_i)-\level(b_i)=y_i$,
whence edge $((a_i)^{i-1}_Y,(b_i)^i_Y)$ is present in $\G_Y\m \phi^{-1}(F)$.
Now for the remaining edges $(x,y)$ of $P$ with $(1+\level(x) - \level(y))=0$,
we have that $(x^i,y^i)$ exists in $\G_Y\m \phi^{-1}(F)$ for every $0 \le i \le t$.
This shows that there is a path from $s_{Y}^0$ to $v_Y^t$ in $\G_Y\m \phi^{-1}(F)$.

Conversely, suppose $Q=(w_1,\ldots,w_L)$ is a path from $s_{Y}^0$ to $v_Y^t$ in $\G_Y\m \phi^{-1}(F)$.
Then $\phi(Q)=((\phi(w_1),\ldots,$ $\phi(w_L))$ is a path from $s$ to $v$ in $G\m F$.
By the definition of graph $\G_Y$, the sequence $X^+(\phi(Q))$ is the same as $Y$.
\end{proof}

We designate $s_{Y}^0$ as the source of $\G_Y$. 
We compute a fault-tolerant reachability preserver $\H_Y$ of $\G_\Y$ that satisfies 
the condition that upon failure of any set $\F\subseteq E(\G_Y)$ of at most $f(k+1)$ edges,
the vertices reachable from $s_Y^0$ in $\G_Y\m \F$ is identical to vertices reachable from 
$s_Y^0$ in $\H_Y\m \F$.
Baswana, Choudhary, and Roditty~\cite{BCR16} showed how to compute, for any graph with $n$ vertices and $m$ edges,
a sparse  $f$-fault-tolerant reachability preserver in time $O(\min\{2^{f},fn\}\cdot mn)$,
in which the in-degree of each vertex is bounded by $2^{f}$.
Thus $\H_Y$ contains at most $(t+1)2^{f(t+1)}n\leq (k+1)2^{f(k+1)}n$ edges,
and is computable in $O(\min\{2^{f(k+1)},fkn\}\cdot k^2 mn)$ time.

For each $Y\in \Y$, we compute a subgraph $H_Y$ of $G$
by just applying projection map $\phi$ over edges of $\H_Y$.
In other words, for each edge $e=(a,b)$ in $G$, we include $e$ in $H_Y$ if and only if $\phi^{-1}(e)\cap E(\H_\Y)$ is non-empty.
Therefore, $H_Y$ contains at most $ (k+1)2^{f(k+1)}n$ edges.
The following lemma is immediate corollary of \autoref{lemma:calH_Y}.

\begin{lemma}
For any triplet $(v,F,Y)$, where  $v\in V$, $F\subseteq E$ has size at most $f$, and $Y=X^+(P_{v,F})$, 
the subgraph $H_Y\m F$ satisfies $d(s,v,H_Y\m F)=d(s,v,G\m F)$.
\label{lemma:H_Y}
\end{lemma}

Finally, let $H$ be a graph obtained by taking union of edges in $H_Y$  for all $y\in \Y$.
There are at most $|\Y|(k+1)2^{f(k+1)}n$ edges in $H$, which is in $2^{O(fk)}n$.
By \Cref{lemma:X_property,lemma:H_Y}, 
and the definition of $\Y$ as the set of positive integer sequences with sum at most $k$,
it follows that $H$ is a $(f,k)$-\normalfont{FT-BFS}.
We obtain the following result.

\begin{theorem}[\autoref{thm:intro_FT-BFS} with explicit exponents]
For any parameters $f,k\geq 1$ and directed graph $G$ with $n$ vertices and $m$ edges, 
we can compute in time $O(\min\{2^{fk+f+k},2^k fkn\} \cdot k^2 mn)$
time an $(f,k)$-\normalfont{EFT-BFS} with $O(2^{fk+f+k} \nwspace k \cdot n)$ edges.
\end{theorem}


We now turn to the design of an $f$-edge fault-tolerant distance sensitivity oracles that, when queried on a pair $(t,F)$, where $v$ is a vertex of $G$ and $F$ is a set of at most $f$ edges of $G$, is able to decide whether $d(s,t,G-F) \leq d(s,t,G)+k$ or not. 
We use the Monte Carlo fault-tolerant reachability oracle $\D$ designed by Brand and Saranurak~\cite{BS19} for directed graphs. This oracle, when queried on a triple $(s,t,F)$, where $s$ and $t$ are two vertices of $G$ and $F$ is a subset of vertices and edges of $G$, reports whether or not $t$ is reachable from $s$ in $G-F$ in $O(|F|^\omega)$ time.
The preprocessing time of $\D$ is $O(n^\omega)$ and the space\footnote{%
	In~\cite{BS19}, this is phrased as $O(n^2 \log n)$ \emph{bits} of space.
} 
in $O(n^2)$. We must note here that $\D$ is a randomized Monte Carlo oracle and the answer to each query is correct with high probability, i.e. the oracle
may return incorrect answer with probability at most $1/n^c$,
for any constant $c$.

We compute such an oracle for each of the graphs $\G_Y$, $Y\in \Y$, 
that is resilient to up to $f(k+1)$ failures. 
The time to compute $\D(G_Y)$ for each $y$ is 
$O(|\Y| (kn)^\omega) = O(2^k (kn)^\omega)$.
The total size of our data structure is $O(2^k (kn)^2)$.
The algorithm to check for a query pair $(v,F)$, 
whether or not the $s$-$v$-distance in $G$ and $G\m F$
differ by an additive term of at most $k$ is as follows. 
For each $Y=(y_1,\ldots,y_t)\in \Y$, we use $\D(\G_Y)$ to check if 
there is a path from $s_{Y}^0$ to $v_Y^t$ in $\G_Y\m \phi^{-1}(F)$.
If the answer is negative for all choices of $Y$,
we report that $s$-$v$-distance increases by more than $k$,
otherwise we report that the distances differ by at most $k$.
The query time of our oracle is $O(|\Y|~|\phi^{-1}(F)|^\omega)$
which is at most $O(2^k (fk)^\omega)$.
The correctness follows from Lemma~\ref{lemma:calH_Y}.
Our oracle is randomized Monte Carlo since the reachability oracle by Brand and Saranurak~\cite{BS19} is inherently randomized.
This concludes the proof of \autoref{thm:intro_oracle_stretched_distances} (restated below).

\oraclestretcheddistances*

\bibliographystyle{plainurl} 
\bibliography{F(P)T_bib}

\end{document}